\documentclass[conference, 10pt, letterpaper]{IEEEtran}

\usepackage[utf8]{inputenc} 
\usepackage[T1]{fontenc}
\usepackage{cite}
\usepackage[cmex10]{amsmath} % Use the [cmex10] option to ensure complicance
\usepackage{epsfig}
\usepackage{cite}
\usepackage{amsmath,amssymb,amsfonts}
\usepackage{algorithmic}
\usepackage{graphicx}
\usepackage{textcomp}
\usepackage{amsthm}
\usepackage[bottom]{footmisc}
\def\BibTeX{{\rm B\kern-.05em{\sc i\kern-.025em b}\kern-.08em
    T\kern-.1667em\lower.7ex\hbox{E}\kern-.125emX}}
\newtheorem{theorem}{Theorem}

\newtheorem{lemma}{Lemma}

\begin{document}
\title{Wasserstein GAN Can Perform PCA} 

% %%% Single author, or several authors with same affiliation:
 \author{%
   \IEEEauthorblockN{Jaewoong Cho and Changho Suh}
   \IEEEauthorblockA{School of Electrical Engineering\\
                     Korea Advanced Institute of Science and Technology\\
                     %CH-8092 Zürich, Switzerland\\
                     Email: \{cjw2525,chsuh\}@kaist.ac.kr}
 }

%%% Several authors with up to three affiliations:
%\author{%
%  \IEEEauthorblockN{Stefan M.~Moser}
%  \IEEEauthorblockA{ETH Zürich\\
%                    ISI (D-ITET), ETH Zentrum\\
%                    CH-8092 Zürich, Switzerland\\
%                    Email: moser@isi.ee.ethz.ch}
%  \and
%  \IEEEauthorblockN{Albus Dumbledore and Harry Potter}
%  \IEEEauthorblockA{Hogwarts School of Witchcraft and Wizardry\\
%                    Hogwarts Castle\\ 
%                    1714 Hogsmeade, Scotland\\
%                    Email: \{dumbledore, potter\}@hogwarts.edu}
%}

\maketitle

%%%%%%
%% Abstract: 
%% If your paper is eligible for the student paper award, please add
%% the comment "THIS PAPER IS ELIGIBLE FOR THE STUDENT PAPER
%% AWARD." as a first line in the abstract. 
%% For the final version of the accepted paper, please do not forget
%% to remove this comment!
%%
\begin{abstract}
Generative Adversarial Networks (GANs) have become a powerful framework to learn generative models that arise across a wide variety of domains. While there has been a recent surge in the development of numerous GAN architectures with distinct optimization metrics, we are still lacking in our understanding on how far away such GANs are from optimality. In this paper, we make progress on a theoretical understanding of the GANs under a simple linear-generator Gaussian-data setting where the optimal maximum-likelihood generator is known to perform Principal Component Analysis (PCA). We find that the original GAN by Goodfellow \emph{et. al.} fails to recover the optimal PCA solution. On the other hand, we show that Wasserstein GAN can approach the PCA solution in the limit of sample size, and hence it may serve as a basis for an optimal GAN architecture that yields the optimal generator for a wide range of data settings.  
\end{abstract}

%% The paper must be self-contained. However, if you are referring to
%% a full version for checking certain proofs, please provide the
%% publically accessible location below.  If the paper is completely
%% self-contained, you can remove the following line from your
%% submission.
%\textit{A full version of this paper is accessible at:}
%\url{http://isit2019.fr/} 

\section{Introduction}
\label{section:Introduction}
%Paragraph 1:
The problem of learning the probability distribution of data is one of the most fundamental problems in statistics and machine learning. A generative model plays a crucial role as an underlying framework by providing a functional block (called the generator in the literature) which can create fake data which resembles the distribution  of real data. Generative Adversarial Networks (GANs)~\cite{GAN} have provided very powerful and efficient solutions for learning the generative model. The GAN framework includes two major components: Generator and Discriminator. This is inspired by a two-player game in which one player, Generator, wishes to generate fake samples that are close to real data, while the other player, Discriminator, wants to discriminate real samples against fake ones. Since the firstly introduced GAN~\cite{GAN} (that we call vanilla GAN), there has been a proliferation of GAN architectures with distinct optimization metrics including f-GAN~\cite{fGAN}, MMD-GAN~\cite{MMDGAN,MMDGAN2}, WGAN~\cite{WGAN}, Least-Squares GAN~\cite{LSGAN}, Boundary equilibrium GAN~\cite{BEGAN}, etc. One natural question that arises in this context is: Is there an optimal GAN architecture among such GANs that provides optimal solutions for learning true distributions? More specifically, is there a proper optimization metric among those employed in such GANs that yields a generated distribution which maximizes the likelihood function?

%Paragraph 2:
In an effort to make progress towards answering this question, we take into consideration a simple canonical setting in which the optimal solution for learning distributions is well-known and hence one can readily figure out an optimization metric (if any) that yields such a solution. The simple setting represents the case in which the data has a high-dimensional Gaussian distribution and a generator is subject to a linear operation with a Gaussian input. It has been shown in this benchmark setting that the maximum-likelihood solution performs Principal Component Analysis (PCA)~\cite{PPCA}, i.e., the covariance matrix of the generated distribution takes principal components of the true covariance matrix.  

%Paragraph 3:
The first finding of this work is that vanilla GAN does not recover the PCA solution under the Gaussian setting when there is no constraint in the discriminator. In the discriminator-unconstrained setting, vanilla GAN can be translated into an optimization that minimizes Jensen-Shannon (JS) divergence between the true and generated distributions. Here we find that whenever the rank of the covariance matrix of the generated Gaussian distribution is smaller than that of the true distribution (which is a typical scenario), the JS divergence attains the same value of $\log 2$ regardless of the generator~\cite{WGAN}. Hence, vanilla GAN fails to achieve the PCA solution\footnote{In practice, the discriminator has some constraints (e.g., being subject to a neural-net function). It has been studied in\cite{Arora},\cite{Farzan} that such constraint gives a positive effect in the generator design, thus leading to a reasonably good performance in practice.}.

%Paragraph 4:
Recently such issue on the JS divergence motivated Arjovsky \emph{et. al.}~\cite{WGAN} to propose Wasserstein GAN (WGAN) which replaces the JS divergence with the first-order Wasserstein distance that does not pose the issue. Then, one may wonder if WGAN can recover the PCA solution in the linear Gaussian setting at hand? 
The main contribution of this paper is to show that this is the case. Specifically we prove that WGAN performs PCA in the limit of sample size. To prove this, we first translate WGAN optimization into another equivalent optimization. Exploiting the key condition that the optimal solution of the translated optimization should satisfy~\cite{Ding:06}, we show that the optimal solution enables PCA. Our finding suggests that WGAN may be a good candidate for an optimal GAN architecture that yields an optimal generator for a wide variety of settings beyond the Gaussian case. 

%Paragraph 5:
We also investigate stability issues for two prominent neural-net-based algorithms~\cite{WGAN},~\cite{WGANGP} that intend to achieve the Nash equilibrium promised by the WGAN optimization. The WGAN optimization can be formulated as a minimax optimization which is actually a challenging problem as the convergence to a bad local optima may often occur. We show via empirical results that the two algorithms are actually practically appealing at least for the Gaussian setting. Our empirical simulation reveals that the deep-learning-based algorithms ensure the fast convergence to the Nash equilibrium with a small gap to the optimality. 

{\bf Related work: } Recently, Feizi \emph{et. al.}~\cite{UGAN} explored a natural way of specifying a loss function that leads to a unified GAN architecture. The authors have investigated a quadratic loss based on the second-order Wasserstein distance to show that the corresponding GAN architecture (which they call Quadratic GAN) can perform PCA under the linear Gaussian setting. This suggests that Quadratic GAN may also be a good candidate for an optimal architecture. However, their algorithm for implementing Quadratic GAN is tailored for the Gaussian setting, so the development of generic algorithms that span a wide spectrum of data settings has been out of reach. On the other hand, we focus on WGAN for which generic neural-net-based algorithms have been well established, and promise that WGAN may be more practically appealing towards an optimal GAN architecture. 
\begin{figure}[t]
\begin{center}
{\epsfig{figure=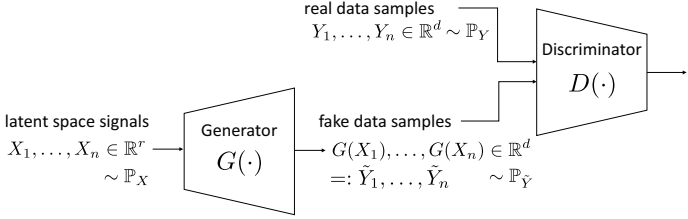, angle=0, width=0.5\textwidth}}
\end{center}
\caption{An architecture for Generative Adversarial Networks (GANs). } \label{fig:model}
\vspace*{-0.1in}
\end{figure}  

\section{GAN Architecture}
Fig. 1 illustrates an architecture for Generative Adversarial Networks (GANs). Suppose we are given the number $n$ of real data samples. We denote those by $({Y}_1, {Y}_2, \dots, {Y}_n)$ where ${Y}_i \in \mathbb{R}^d$ and $d$ indicates a dimension of each sample. Let $\mathbb{P}_{Y}$ denote the probability distribution of ${Y}_i$. Let $(X_1, X_2, \dots, X_n)$ be latent-space signals: inputs to the generator $G(\cdot)$ where $X_i \in \mathbb{R}^r$, and $r$ indicates a dimension of the latent signal, usually $r\leq d$. Let $\tilde{Y}_i:= G({X}_i) \in \mathbb{R}^d$ be the $i$th fake sample and $\mathbb{P}_{\tilde{Y}}$ be the probability distribution of $\tilde{Y}_i$.

We focus on a linear-generator Gaussian-data setting in which $Y_i$'s are i.i.d. each according to $\mathbb{P}_Y = {\cal N} ({\bf 0}, {\bf K}_Y)$; $X_i$'s are i.i.d. $\sim {\cal N}({\bf 0}, {\bf I}_r)$; and $G(\cdot)$ is a linear operator implemented by a matrix ${\bf G} \in \mathbb{R}^{d \times r}$.

The GAN architecture is inspired by a two-player game in which one player is the discriminator $D(\cdot)$ who wishes to discriminate real samples against fake ones; and the other is the generator  $G(\cdot)$ who wants to fool the discriminator. Many GAN approaches have been developed as an effort to achieve such goals, and corresponding optimization problems are formulated with different optimization metrics~\cite{GAN,WGAN,fGAN,MMDGAN,MMDGAN2,BEGAN,LSGAN}.

\section{Vanilla GAN}
\label{section:VGAN}
In an effort to achieve such goals, Goodfellow \emph{et. al.}~\cite{GAN} came up with an insightful interpretation: Viewing $D(\cdot)$ as the \emph{probability} that the input (taking either a real sample $Y_i$ or a fake one $\tilde{Y}_i$) is a real sample. This viewpoint motivated them to introduce the following minimax optimization. Given $(X_1, Y_1), \dots, (X_n,Y_n)$:
\begin{align*}
\min_{G(\cdot)} \max_{D(\cdot)} 
\frac{1}{n} \sum_{i=1}^n \log D(Y_i) + \frac{1}{n} \sum_{i=1}^n \log ( 1 - D(\tilde{Y}_i)).
\end{align*}
Notice that the discriminator intends to maximize the logarithmic of such probability $D(\cdot)$ when the input is real (or the logarithm of ``$1 - \textrm{such probability}$'' $1- D(\cdot)$ for a fake sample).

For simplicity, we consider a regime in which the sample size $n$ is large enough that the objective function in the above can be well approximated as the population limit. So the optimization can be approximated as: 
\begin{equation}
\label{equation:minmax}
\min_{G(\cdot)}\max_{D(\cdot)}\ \mathbb{E}_{{Y}}[\log D({Y})]+\mathbb{E}_{\tilde{{Y}}}[\log(1-D(\tilde{{Y}}))].
\end{equation}

In this work, we find that vanilla GAN does not recover the optimal PCA solution even for the simple Gaussian setting. To see this, we first rewrite the objective function as: 
\begin{align}
\int_{z \in {\cal Y} \cup {\cal \tilde{Y}}}\left[{\mathbb P}_{{Y}}({z})\log D({z})+{\mathbb{P}}_{{\tilde{Y}}}({z})\log(1-D({z}))\right]d{z},
\end{align}
where $\cal {Y}$ and ${\cal \tilde{Y}}$ indicate the ranges of $Y$ and $\tilde{Y}$, respectively. Here we set:
\begin{align*}
&\mathbb{P}_Y (z) dz := 0 \quad \textrm{if }z \in {\cal \tilde{Y}} \setminus {\cal Y}; \\
& \mathbb{P}_{\tilde{Y}} (z) dz := 0  \quad \textrm{if }z \in  {\cal Y} \setminus  {\cal \tilde{Y}}.
\end{align*}
Observe that for any $(a,b)\in\mathbb{R}^2\backslash\{0,0\}$ and $t \in (0,1)$, the function $a\log t+b\log(1-t)$ is maximized at $t^*=\frac{a}{a+b}$. Hence, for a fixed $G(\cdot)$, the optimal discriminator reads:
\begin{align*}
D^{*}({z}) = \frac{  \mathbb{P}_Y (z) }{ \mathbb{P}_Y (z) + \mathbb{P}_{\tilde{Y}} (z) }.
\end{align*}
Plugging this into~(\ref{equation:minmax}), we get:
\begin{align*}
%\label{equation:minmax}
&\min_{\bf G}\mathbb{E}_{{Y}}[\log D^*({Y})]+\mathbb{E}_{{\tilde{Y}}}[\log(1-D^*(\tilde{Y}))]\\
=&\min_{\bf G}\mathbb{E}_{{Y}}\left[\log \frac{\mathbb{P}_{{Y}}({Y})}{\mathbb{P}_{{\tilde{Y}}}({Y})+\mathbb{P}_{\tilde{{Y}}}({Y})}\right]+\mathbb{E}_{\tilde{Y}}\left[\log\frac{\mathbb{P}_{\tilde{{Y}}}({\tilde Y})}{\mathbb{P}_{{Y}}({\tilde Y})+\mathbb{P}_{\tilde{{Y}}}({\tilde Y})}\right]\\
=&\min_{\bf G}\ {\sf KL}\left(\mathbb{P}_{{Y}}\left\|\frac{\mathbb{P}_{{Y}}+\mathbb{P}_{\tilde{{Y}}}}{2}\right.\right)
+{\sf KL}\left(\mathbb{P}_{\tilde{{Y}}}\left\|\frac{\mathbb{P}_{{Y}}+\mathbb{P}_{\tilde{{Y}}}}{2}\right.\right)
-\log4\\
=&\min_{\bf G}\ 2\cdot {\sf JSD}\left(\mathbb{P}_{{Y}}\left\|\mathbb{P}_{\tilde{{Y}}}\right.\right)
-\log4,
\end{align*} 
where ${\sf KL}(\cdot \| \cdot)$ and ${\sf JSD} (\cdot \| \cdot)$ indicate the KL and JS divergences, respectively~\cite{CoverThomas, JSD}: For distributions $p$ and $q$, 
\begin{align*}
{\sf JSD}(p \| q):= \frac{1}{2}
\left( {\sf KL}\left(p \left\| \frac{p+q}{2}\right.\right)+{\sf KL}\left(q \left\| \frac{p+q}{2}\right.\right)\right).
\end{align*}
In the linear Gaussian setting, ${\tilde{Y}}$ is also Gaussian ${\cal N}({\bf 0}, {\bf K}_{\tilde{Y}})$ where $\mathbf{K}_{\tilde{Y} } = \mathbb{E} [ {\bf G} X ({\bf G}X)^T ] = {\bf G} {\bf G}^T$. Typically the dimension of the latent signals is smaller than than of data. So in this case, the support of $\mathbb{P}_{\tilde{Y}}$ has a strictly lower dimension than that of $\mathbb{P}_Y$, yielding:
\begin{align}
\frac{\mathbb{P}_Y (z) + \mathbb{P}_{\tilde{Y}} (z) }{2}  dz = 
\begin{cases}
\frac{\mathbb{P}_Y (z)}{2} dz \textrm{ if } z \in {\cal Y} \setminus {\cal \tilde{Y} }; \\
\frac{\mathbb{P}_{\tilde{Y}} (z) }{2} dz
\textrm{ if } z \in {\cal \tilde{Y}} \setminus {\cal Y}.
\end{cases}
\end{align}
This then gives 
$
{\sf KL} \left(\mathbb{P}_Y \left\| \frac{\mathbb{P}_Y  + \mathbb{P}_{\tilde{Y}}  }{2}\right.\right) = {\sf KL} \left(\mathbb{P}_{\tilde{Y}} \left\| \frac{\mathbb{P}_{Y}  + \mathbb{P}_{\tilde{Y}}  }{2}\right.\right) = \log 2,
$
which in turn yields:
\begin{equation}
\label{equation:JSDlog2}
{\sf JSD} (\mathbb{P}_Y \| \mathbb{P}_{\tilde{Y}} ) = \log 2,
\end{equation}
regardless of how we design $\mathbf{G}$. This implies that an optimal $\mathbf{G}^*$ does not necessarily perform PCA.    

\section{Wasserstein GAN}
\label{section:WGAN}    
The main reason that vanilla GAN fails to recover the optimal solution is that it is based on the JS divergence which poses the critical issue reflected in~(\ref{equation:JSDlog2}). As an effort to avoid such an issue, one may consider another prominent GAN architecture developed by Arjovsky \emph{et. al.}~\cite{WGAN}: Wasserstein GAN (WGAN) which employs the first-order Wasserstein distance instead of the JS divergence. The first-order Wasserstein distance is defined as: Given two distributions $\mathbb{P}_{{Y}}$ and $\mathbb{P}_{\tilde{{Y}}}$,
\begin{align}
\label{eq:W1Def}
W(\mathbb{P}_{{Y}},\mathbb{P}_{\tilde{{Y}}}) 
:= \min_{\mathbb{P}_{{Y},\tilde{{Y}}}}\mathbb{E}
\left[\|{Y}-\tilde{{Y}}\|\right],
\end{align}    
where the minimization is over all joint distributions which respect the marginals $\mathbb{P}_{{Y}}$ and $\mathbb{P}_{\tilde{{Y}}}$. Here $\| \cdot \|$ indicates the $\ell_2$ norm. Unlike  the JS divergence, it provides a meaningful non-saturating value even when dimensions of supports of two distributions are distinct. WGAN intends to solve the following optimization: 
\begin{align}
\label{equation:WGANopt}
\min_{\mathbf{G}}W(\mathbb{P}_{{Y}},\mathbb{P}_{\tilde{{Y}}}).
\end{align} 

The main contribution of this paper is that unlike vanilla GAN, WGAN recovers the optimal PCA solution under the linear Gaussian setting in the limit of sample size, formally stated below. 
\begin{theorem}[]
\label{theorem:W1PCA}
Let ${Y}\sim\mathcal{N}({\bf 0},{\bf K}_{{Y}})$ where ${\bf K}_{{Y}}$ has a full rank, i.e., ${\sf rank}({\bf K}_Y) = d$. Let ${X}\sim\mathcal{N}(\mathbf{0},\mathbf{I}_{r})$ where $r\leq d$. The optimal solution of the WGAN optimization~(\ref{equation:WGANopt}) under a linear generator is the $r$-PCA solution.
\end{theorem}
\begin{proof}
The proof consists of two parts:
\begin{itemize}
\item[(a)] The WGAN optimization~\eqref{equation:WGANopt} is translated into another equivalent optimization which was investigated in depth in prior works~\cite{Ding:06};  
\item[(b)] We show that the $r$-PCA solution respects the unique condition (derived in~\cite{Ding:06}) that the optimal solution of the translated optimization satisfies.  
\end{itemize} 

By the definition of the Wasserstein distance~(\ref{eq:W1Def}) and the assumption of a linear generator, 
\begin{align}
\label{eq:WGANoptEq}
\min_{{\bf G}}W(\mathbb{P}_{{Y}},\mathbb{P}_{{\bf G}{\rm X}})
&=\min_{{\bf G}}\min_{\mathbb{P}_{{Y},{\bf G}{X}}}\mathbb{E}\left[\|{{Y}}-{\bf G}{{{X}}}\|\right].
\end{align}

Lemma 1 below casts the WGAN optimization~\eqref{eq:WGANoptEq} into another equivalent optimization. 
\begin{lemma}
\label{lemma:W1GANeqR1PCA}
Under a linear generator, 
\begin{align}
\label{equation:L1optimization}
\min_{{\bf G}}\min_{\mathbb{P}_{{Y},{\bf G}X}}\mathbb{E}\left[\|{{Y}}-{\bf G}{{{X}}}\|\right]
= \min_{{\bf U}^T{\bf U}= \mathbf{I}_r}\mathbb{E}\left[\|{Y}-{\bf U}{\bf U}^{T}{Y}\|\right].
\end{align} 
\end{lemma}
\begin{proof}
See Section~\ref{sec:W1GANeqR1PCA}.
\end{proof}
\emph{Remark:} Lemma~\ref{lemma:W1GANeqR1PCA} suggests that for the optimal $\mathbf{G}^*$ and $\mathbf{U}^*$, we have 
$\mathbf{G}^* X = \mathbf{U}^{*} \mathbf{U}^{*T} Y$. Considering the covariances of these quantities, we see that
\begin{align}
\label{eq:GUcondition}
\mathbf{G}^* \mathbf{G}^{*T} = 
\mathbf{U}^* \mathbf{U}^{*T} \mathbf{K}_Y \mathbf{U}^* \mathbf{U}^{*T}.
\end{align}
$\Box$

For the translated optimization~(\ref{equation:L1optimization}),~\cite{Ding:06} derived the key condition that the optimal solution $\mathbf{U}^*$ should satisfy. This condition turns to play a crucial role to prove the theorem. We first introduce a notation which serves describing the condition. Let
\begin{align}
\label{eq:keymatrix}
\mathbf{M} = \mathbb{E} 
\left[\frac{YY^T}{\|Y-{\bf U}^*{\bf U}^{*T}Y\|}\right ].
\end{align}
Now the key condition w.r.t. $\mathbf{U}^*$ says: 
\begin{align}
\label{eq:keycondition}
{\bf U}^* = r\textrm{-}{\sf PrincipalEigenvectors}(\mathbf{M}),
\end{align}
which means a concatenation of the $r$ principal eigenvectors of $\mathbf{M}$. It was also shown in~\cite{Ding:06} that $\mathbf{U}^*$ satisfying the key condition~\eqref{eq:keycondition} is \emph{unique}.
Here what we prove is that such $\mathbf{U}^*$ satisfies (see below for the proof):
\begin{align}
\label{eq:keyproof}
\mathbf{U}^* = r\textrm{-}{\sf PrincipalEigenvectors}(\mathbf{K}_Y).
\end{align} 
Applying \eqref{eq:keyproof} into \eqref{eq:GUcondition}, we can find the optimal $\mathbf{G}^*$ satisfies: 
\begin{align}
\mathbf{G}^* \mathbf{G}^{*T} = 
\mathbf{V} {\sf diag}(\sigma_1^2, \dots, \sigma_r^2, 0, \dots, 0) \mathbf{V}^T
\end{align}
where $\mathbf{V}:=[\mathbf{v}_1,\dots,\mathbf{v}_d]\in \mathbb{R}^{d \times d}$ and $\sigma_i^2$ are eigen-components of $\mathbf{K}_Y$:
\begin{align}
\label{eq:Kycov}
\mathbf{K}_Y = \mathbf{V}{\sf diag}(\sigma_1^2, \dots, \sigma_d^2) \mathbf{V}^T
=: \mathbf{V} \mathbf{\Sigma} \mathbf{V}^T. 
\end{align}
Here $\mathbf{v}_i$'s are orthonormal vectors and $\sigma_1 \geq \cdots \geq \sigma_d$.
This implies that the optimal generator $\mathbf{G}^*$ is formed by taking the $r$ principal components of $\mathbf{K}_Y$. This completes the proof. From below, we will prove the main claim~\eqref{eq:keyproof}. 

\emph{Proof of~\eqref{eq:keyproof}}: Since ${\bf V}$ is an orthonormal matrix, we get:
\begin{align*}
{Y}
= {\bf V}{\bf V}^T{Y}
= \sum_{i=1}^d Z_i\mathbf{v}_i,
\end{align*} 
where $Z_i:=\mathbf{v}_i^T{Y}$. Since ${Y}\sim{\cal N}({\bf 0},{\bf V}{\bf \Sigma}{\bf V}^T)$~\eqref{eq:Kycov},
\begin{align}
\label{eq:Zintroduce}
{Z}:={\bf V}^T{Y}=[Z_1;\cdots;Z_d]\sim\mathcal{N}(\mathbf{0},{\bf \Sigma}).
\end{align}
%with mean $\mathbf{0}$ and the covariance matrix $\Sigma(=V^TK_YV)$.

Let ${\bf V}_r := [\mathbf{v}_1,\dots,\mathbf{v}_r]\ (r$-principal eigenvectors of $\mathbf{K}_{Y}$). Now we intend to show that $\mathbf{V}_r = \mathbf{U}^*$. To this end, it suffices to show that $\mathbf{V}_r$ satisfies the key condition of (10) due to the uniqueness of $\mathbf{U}^*$. So, we plug ${\bf V}_r$ into~\eqref{eq:keymatrix} (by replacing $\mathbf{U}^*$), thus obtaining:
\begin{align}
&\mathbf{M}=\mathbb{E}\left[ \frac{YY^T}{\|{Y}-{\bf V}_r{\bf V}_r^T{Y}\|}\right]
= \mathbb{E}\left[ \frac{YY^T}{\sqrt{\|\sum_{i=r+1}^d Z_i\mathbf{v}_i}\|^2}\right]\nonumber\\
&\overset{(a)}{=} \mathbb{E}\left[ \frac{YY^T}{\sqrt{\sum_{i=r+1}^d Z_i^2}}\right]=
\mathbb{E}\left[
\frac{\left(\sum_{j=1}^dZ_j\mathbf{v}_j\right)
\left(\sum_{k=1}^d Z_k\mathbf{v}_k^T\right)}{\sqrt{\sum_{i=r+1}^d Z_i^2}}
\right]\nonumber\\
&=
\sum_{j=1}^d\mathbb{E}\left[
\frac{Z_j^2}{\sqrt{\sum_{i=r+1}^d Z_i^2}}\right]\mathbf{v}_j\mathbf{v}_j^T
+ 
\sum_{j\neq k, j<k}2\mathbb{E}\left[
\frac{Z_jZ_k}{\sqrt{\sum_{i=r+1}^d Z_i^2}}\right]\mathbf{v}_j\mathbf{v}_k^T\nonumber\\
&\overset{(b)}{=} 
\sum_{j=1}^d\mathbb{E}\left[
\frac{Z_j^2}{\sqrt{\sum_{i=r+1}^d Z_i^2}}\right]\mathbf{v}_j\mathbf{v}_j^T,\label{eq:1}
\end{align}
where $(a)$ follows from the orthonormality of $\mathbf{v}_i$'s and $(b)$ follows from the fact that 
\begin{align}
\label{eq:jk}
\mathbb{E}\left[
\frac{Z_jZ_k}{\sqrt{\sum_{i=r+1}^d Z_i^2}}\right]=0,
\end{align}
which we will show in the sequel. 

\emph{Proof of~\eqref{eq:jk}}: Let ${Z}^r: = \{Z_{r+1},\ldots,Z_{d}\}$. Using the tower property and the fact that ${Z}=[Z_1;\ldots;Z_d]\sim\mathcal{N}(\mathbf{0},{\bf \Sigma})$ (see~\eqref{eq:Zintroduce}), we get
\begin{align*}
&\mathbb{E}\left[
\frac{Z_jZ_k}{\sqrt{\sum_{i=r+1}^d Z_i^2}}\right]
=\mathbb{E}_{{Z}^r,Z_k}\left[\mathbb{E}_{Z_j}\left[\left.
\frac{Z_jZ_k}{\sqrt{\sum_{i=r+1}^d Z_i^2}}
\right|{Z}^r,Z_k
\right]
\right]\\
&\overset{(a)}{=}\mathbb{E}_{{Z}^r,Z_k}\left[\mathbb{E}_{Z_j}\left[\left.
\frac{-Z_jZ_k}{\sqrt{\sum_{i=r+1}^d Z_i^2}}
\right|{Z}^r,Z_k
\right]
\right]\\
&=-\mathbb{E}\left[
\frac{Z_jZ_k}{\sqrt{\sum_{i=r+1}^d Z_i^2}}\right],
\end{align*}
where $(a)$ follows from the fact that $Z_j\sim\mathcal{N}(0,\sigma_j^2)$ has a symmetric pdf and $\frac{Z_jZ_k}{\sqrt{\sum_{i=r+1}^d Z_i^2}}$ is an odd function of $Z_j$.
Therefore,
\begin{align*}
\mathbb{E}\left[
\frac{Z_jZ_k}{\sqrt{\sum_{i=r+1}^d Z_i^2}}\right]=0.
\end{align*} 
Computing $\mathbf{M} \mathbf{v}_k$ with the help of~\eqref{eq:1}, we get:
\begin{align*}
\mathbf{M} \mathbf{v}_k
\overset{}{=} 
\mathbb{E}\left[
\frac{Z_k^2}{\sqrt{\sum_{i=r+1}^d Z_i^2}}\right]\mathbf{v}_k.
\end{align*}
This implies that $\mathbf{v}_k$ is an eigenvector of $\mathbf{M}$ and the corresponding eigenvalue is
$
\mathbb{E}\left[
\frac{Z_k^2}{\sqrt{\sum_{i=r+1}^d Z_i^2}}\right].
$ 

Now it suffices to show that $(\mathbf{v}_1, \dots, \mathbf{v}_r)$ are $r$-principal eigenvectors of $\mathbf{M}$. To show this, we will demonstrate below that for $j\leq r$ and $k>r$,
\begin{align}
\label{eq:principalvalue}
\mathbb{E}\left[\frac{Z_j^2}{\sqrt{\sum_{i=r+1}^dZ_i^2}}\right]
\geq
\mathbb{E}\left[
\frac{Z_k^2}{\sqrt{\sum_{i=r+1}^d Z_i^2}}
\right].
\end{align}
Since $j \leq r$ and $Z_i$'s are independent, we have:
\begin{align*}
&\mathbb{E}\left[\frac{Z_j^2}{\sqrt{\sum_{i=r+1}^dZ_i^2}}\right]
= \mathbb{E}\left[Z_j^2\right]\mathbb{E}\left[\frac{1}{\sqrt{\sum_{i=r+1}^dZ_i^2}}\right]\\
&\overset{(a)}{\geq}
\mathbb{E}\left[Z_k^2\right]
\mathbb{E}\left[
\frac{1}{\sqrt{Z_k^2+\sum_{i>r, i\neq k}Z_i^2}}
\right]
\\
&\overset{(b)}{=}
\mathbb{E}_{Z_k}\left[Z_k^2\right]
\mathbb{E}_{{Z}^r}\left[
\mathbb{E}_{Z_k}\left[\left.
\frac{1}{\sqrt{Z_k^2+\sum_{i>r, i\neq k}Z_i^2}}
\right|
{Z}^r
\right]
\right]\\
&\overset{(c)}{\geq}
\mathbb{E}_{Z_k}\left[Z_k^2\right]
\mathbb{E}_{{Z}^r}\left[
\frac{1}{\sqrt{\mathbb{E}_{Z_k}\left[Z_k^2\right]+\sum_{i>r, i\neq k}Z_i^2}}
\right]\\
&\overset{}{=}
\mathbb{E}_{{Z}^r}\left[
\frac{\mathbb{E}_{Z_k}\left[Z_k^2\right]}{\sqrt{\mathbb{E}_{Z_k}\left[Z_k^2\right]+\sum_{i>r, i\neq k}Z_i^2}}
\right]\\
&\overset{(d)}{\geq}
\mathbb{E}_{{Z}^r}\left[\mathbb{E}_{Z_k}\left[\left.
\frac{Z_k^2}{\sqrt{Z_k^2+\sum_{i>r, i\neq k}Z_i^2}}
\right|{Z}^r
\right]
\right]\\
&\overset{}{=}
\mathbb{E}\left[
\frac{Z_k^2}{\sqrt{\sum_{i=r+1}^d Z_i^2}}
\right],
\end{align*}
where $(a)$ follows from $\mathbb{E}\left[Z_j^2\right]=\sigma_j^2\geq\sigma_k^2= \mathbb{E}\left[Z_k^2\right]\ (j<k)$; $(b)$ follows from the tower property and ${Z}\sim\mathcal{N}({\bf 0},{\bf \Sigma})$; $(c)$ follows from Jensen's inequality and the fact that $f(x) = \frac{1}{\sqrt{x+a}}$ is a convex function of $x\geq 0$; and $(d)$ follows from Jensen's inequality and the fact that $f(x) = \frac{x}{\sqrt{x+a}}$ is a concave function of $x\geq 0$. This proves~\eqref{eq:principalvalue}.

\end{proof}
\subsection{Proof of Lemma~\ref{lemma:W1GANeqR1PCA}}
\label{sec:W1GANeqR1PCA}
Let $\tilde{{Y}}={\bf G}{X}$ be a random vector whose support lies in $\mathcal{S}$ ($r$-dimensional subspace in $\mathbb{R}^d$). Then we can write ${Y}$ as:
\begin{align}
{Y} = {Y}_{{\cal S}'} + {Y}_{\cal S},
\end{align}
where ${Y}_{{\cal S}}$ indicates the projection of ${Y}$ onto $\mathcal{S}$. Using this,
\begin{align*}
&\min_{\bf G}W(\mathbb{P}_{{Y}},\mathbb{P}_{\tilde{{Y}}})
=\min_{\bf G}\min_{\mathbb{P}_{{Y},\tilde{{Y}}}}\mathbb{E}[\|{Y}-\tilde{{Y}}\|]\\
&=\min_{\bf G}\min_{\mathbb{P}_{{Y},\tilde{{Y}}}}\mathbb{E}\left[\sqrt{\|{Y}_{\mathcal{S}'}+{Y}_{\mathcal{S}}-\tilde{{Y}}\|^2}\right]\\
&\overset{(a)}{=}\min_{\bf G}\min_{\mathbb{P}_{{{Y}},\tilde{{{Y}}}}}\mathbb{E}\left[\sqrt{\|{Y}_{\mathcal{S}'}\|^2+\|{Y}_{\mathcal{S}}-\tilde{{Y}}\|^2}\right],
\end{align*}
where $(a)$ follows from the orthogonality between ${Y}_{\mathcal{S}'}$ and $({Y}_{\mathcal{S}}-\tilde{{Y}})$. 
Let $\mathbb{P}_{Y,\hat{Y}}^*$ denote a joint distribution such that $Y_{\mathcal{S}}=\hat{Y}$. Then
\begin{align*}
&\mathbb{E}_{\mathbb{P}_{Y,\hat{Y}}^*}\left[\sqrt{\|{Y}_{\mathcal{S}'}\|^2+\|{Y}_{\mathcal{S}}-\tilde{{Y}}\|^2}\right]\\
=&\mathbb{E}_{\mathbb{P}_{Y}}\left[\sqrt{\|{Y}_{\mathcal{S}'}\|^2}\right]
\geq \min_{\mathbb{P}_{{{Y}},\tilde{{{Y}}}}}\mathbb{E}\left[\sqrt{\|{Y}_{\mathcal{S}'}\|^2+\|{Y}_{\mathcal{S}}-\tilde{{Y}}\|^2}\right]. 
\end{align*}
From the fact that $\sqrt{\|Y_{\mathcal{S}'}\|^2}\leq\sqrt{\|Y_{\mathcal{S}'}\|^2+\|Y_{\mathcal{S}}-\hat{Y}\|^2}$,
\begin{align*}
&\mathbb{E}_{\mathbb{P}_{Y,\hat{Y}}^*}\left[\sqrt{\|{Y}_{\mathcal{S}'}\|^2+\|{Y}_{\mathcal{S}}-\tilde{{Y}}\|^2}\right]\\=&\mathbb{E}_{\mathbb{P}_{Y}}\left[\sqrt{\|{Y}_{\mathcal{S}'}\|^2}\right]\leq \min_{\mathbb{P}_{{{Y}},\tilde{{{Y}}}}}\mathbb{E}\left[\sqrt{\|{Y}_{\mathcal{S}'}\|^2+\|{Y}_{\mathcal{S}}-\tilde{{Y}}\|^2}\right]. 
\end{align*}
Therefore, 
\begin{align*}
&\min_{\bf G}W(\mathbb{P}_{{Y}},\mathbb{P}_{\tilde{{Y}}})\\
&\overset{}{=}\min_{\bf G}\min_{\mathbb{P}_{{{Y}},\tilde{{{Y}}}}}\mathbb{E}\left[\sqrt{\|{Y}_{\mathcal{S}'}\|^2+\|{Y}_{\mathcal{S}}-\tilde{{Y}}\|^2}\right]\\
&\overset{}{=}\min_{\bf G}\mathbb{E}_{\mathbb{P}_Y}\left[\sqrt{\|{Y}_{\mathcal{S}'}\|^2}\right]
{=}\min_{\mathcal{S}}\mathbb{E}_{\mathbb{P}_Y}\left[\|{Y}-{Y}_{\mathcal{S}}\|\right]\\
&{=}\min_{\mathbf{U}^T\mathbf{U}=\mathbf{I}_r}\mathbb{E}_{\mathbb{P}_Y}\left[\|{Y}-{\bf U}{\bf U}^T{Y}\|\right], 
\end{align*}
where $\textrm{range}(\mathbf{G}) =\textrm{range}(\mathbf{U}) $.
\begin{figure}[t]
\begin{center}
{\epsfig{figure=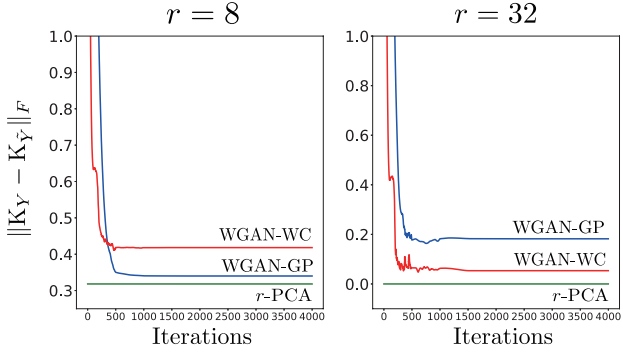, angle=0, width=0.45\textwidth}}
\end{center}
\caption{The performance of WGAN-GP and WGAN-WC for the linear Gaussian setting in which the data dimension $d=32$ and latent signal dimension $r \in \{8, 32\}$. As a performance metric, we employ the Frobenius norm between the ground-truth covariance and the generated covariance. As a baseline, we plot the performance of the $r$-PCA solution.}
\label{fig:experimental_results}
\vspace*{-0.1in}
\end{figure}

\section{Neural-net-based WGAN Algorithms}

So far we have shown that WGAN recovers the optimal PCA solution for the linear Gaussian setting. This suggests that WGAN can be a good candidate for the optimal GAN architecture.  

In practice, however, it is early to arrive at this conclusion. To implement the WGAN solution, usually we hinge upon the Kantorovich dual~\cite{OT:09} of the primal optimization, which includes two function optimizations that can be efficiently solved via neural networks with high accuracy: 
\begin{align*}
\min_{\mathbf{G}}W(\mathbb{P}_{{Y}},\mathbb{P}_{{\bf G}{X}})
=\min_{\mathbf{G}}\max_{\|D\|_L\leq 1}\mathbb{E}_{{Y}}[D({Y})]-\mathbb{E}_{{\tilde Y}}[D({\tilde Y})],
\end{align*}
where $\|D\|_L\leq 1$ means 1-Lipschitz functions $D$: For all $y_1$ and $y_2$, $
|D(y_1)-D(y_2)|\leq \|y_1-y_2\|.
$
But here an issue arises: The Kantorovich dual is of \emph{minimax} optimization where the convergence to a bad local optima frequently occurs. Hence, we also need to investigate the stability of neural-net-based algorithms that implement the minimax solution.

\begin{figure}[t]
\begin{center}
{\epsfig{figure=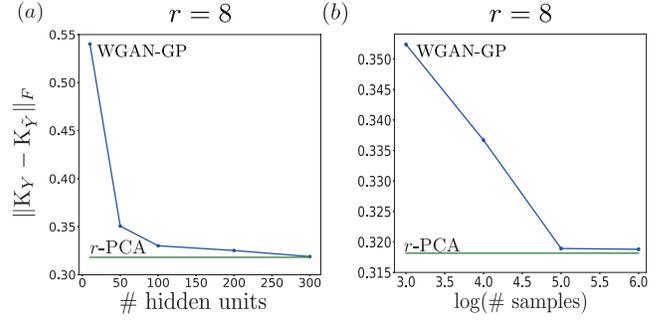, angle=0, width=0.5\textwidth}}
\end{center}
\caption{The performance of WGAN-GP for the linear Gaussian setting in which the data dimension $d=32$ and latent signal dimension $r=8$ in different network sizes with different sample sizes. The discriminator of WGAN-GP is implemented with a neural network composed of five layers in which hidden layers include the same number of hidden units: (a) the performance of WGAN-GP with the sample size $n=100,000$ in different sizes of hidden units $
n_h\in\{10,50,100,200,300\}$; (b) the performance of WGAN-GP with different sample sizes $
n\in\{1000,10000,100000,1000000\}$ for the case of $n_h=300$. We found that the gap reduces when a more complex neural network is employed together with a larger $n$.} 
\label{fig:experimental_results2}
\vspace*{-0.1in}
\end{figure}

In this work, we consider such two prominent algorithms: (1) WGAN with weight clipping (WGAN-WC); (2) WGAN with gradient penalty (WGAN-GP).

Another contribution of this work is to empirically show that the two algorithms yield good stability for the Gaussian setting, i.e., ensure the fast convergence to the Nash equilibrium with a small gap to the optimality.

An empirical verification is done for the following setting. We generate $n \in\{1000, 10000, 100000, 1000000\}$ i.i.d. samples from an $d$-dimensional Gaussian distribution ${Y}\sim\mathcal{N}(\mathbf{0},{\bf K}_{Y})$. We set $d=32$. A generation for $\mathbf{K}_Y$ takes the following procedure. First we generate $\mathbf{V}{\bf \Sigma} {\bf V}^T$ such that $[\mathbf{V}]_{ij}\sim\mathcal{N}(0,1)$ and $\sigma_i^2\sim \text{Uniform}(0,10)$. Next, we normalize $\mathbf{V}{\bf \Sigma }{\bf V}^T$ so as to have the unit Frobenius norm. The generator is constructed with a single matrix $\mathbf{G}$. The discriminator is implemented with a fully connected neural network. As a performance metric, we employ the Frobenius norm between the ground-truth covariance and the generated covariance. We use the batch size of 200. For WGAN-GP, we use Adam optimizer to train neural networks and momentum parameters $(\beta_1,\beta_2)=(0.5,0.9)$, and we set $\lambda$ parameter to be 0.1. For WGAN-WC, we use RMSProp optimizer, and we set $c$ parameter to be 0.01.  

Fig.~\ref{fig:experimental_results} shows the performance of WGAN-GP (blue-colored) and WGAN-WC (red-colored) for the values of $r\in\{8,32\}$ and $n=100,000$ with the Glorot initialization~\cite{Glorot:10}. Here the discriminator is implemented with a neural network composed of three layers, each with 64 neurons and ReLU activation functions. We use $10^{-3}$ as the initial learning rate, and decay it by a factor of 10 at the end of every 5 epochs for both algorithms. For comparison, we also plot the performance of the $r$-PCA with a green-colored line. Observe that both WGAN-GP and WGAN-WC exhibit reasonably fast convergence speed\footnote{Training of WGAN takes less than 60 seconds on a TESLA\_P40 GPU.} (converge only with a few hundred iterations) with a small gap to the optimality, reflected in the green-colored $r$-PCA solution. Here the small gap comes from a finite sample size $n$, as well as imperfect representability of a specific neural network employed herein. Actually we found that the gap reduces when a more complex neural network is employed together with a larger $n$. Here are details - also see Figs.~\ref{fig:experimental_results2} and~\ref{fig:experimental_results3}. 

\begin{figure}[t]
\begin{center}
{\epsfig{figure=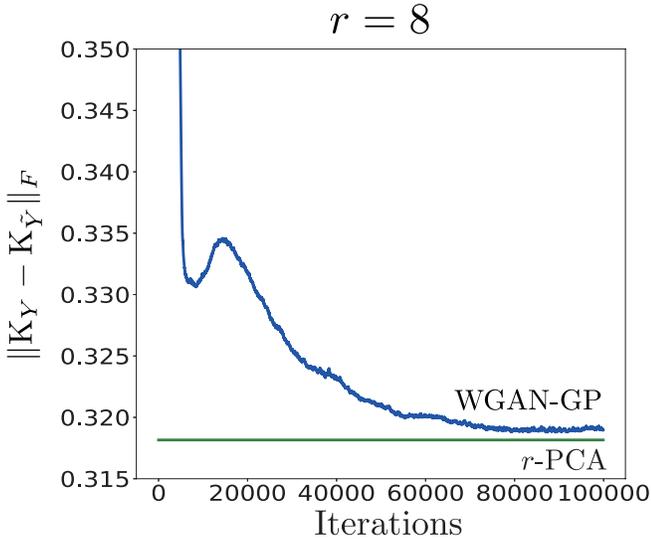, angle=0, width=0.5\textwidth}}
\end{center}
\caption{The performance of WGAN-GP for the linear Gaussian setting in which the data dimension $d=32$ and latent signal dimension $r=8$ with the sample size $n=1,000,000$. The discriminator is implemented with a neural network composed of five layers, each with 300 neurons and ReLU activation functions. We can observe a small gap in the convergence.} 
\label{fig:experimental_results3}
\vspace*{-0.1in}
\end{figure} 

 As a setting that incorporates a more complex network, we consider the case which $r=8$, $n=100,000$, and the number of hidden units for each layer $n_h\in\{10,50,100,200,300\}$; see Fig.~\ref{fig:experimental_results2}$(a)$. Here we employ WGAN-GP. The discriminator is implemented with a neural network composed of five layers in which hidden layers include the same number of hidden units and we use Adam optimizer with the learning rate $10^{-4}$. We see from Fig.~\ref{fig:experimental_results2}$(a)$ that the gap reduces with an increasing complexity, converging almost to 0. Fig.~\ref{fig:experimental_results2}($b$) demonstrates the performance of WGAN-GP with different a sample size $n\in\{1000, 10000, 100000, 1000000\}$ for the case of $n_h=300$. We also see a very small gap for a large sample size. Fig.~\ref{fig:experimental_results3} shows the performance as a function of iterations when $n_h=300$ and the sample size $n=1,000,000$. We can also see a small gap in the convergence.

\begin{figure}[t]
\begin{center}
{\epsfig{figure=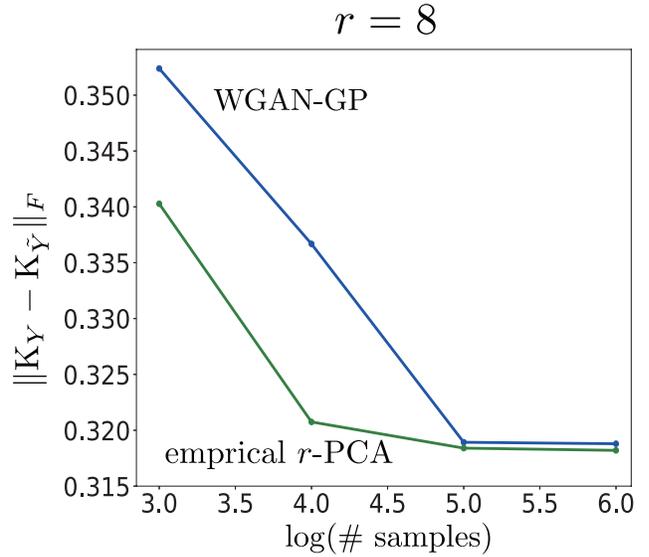, angle=0, width=0.5\textwidth}}
\end{center}
\caption{The comparison between performances of WGAN-GP and the empirical $r$-PCA solution while varying a sample size $n\in\{1000, 10000, 100000, 1000000\}$. We can observe a non-significant gap to the empirical $r$-PCA solution for each sample size and the gap vanishes with an increase in $n$.} 
\label{fig:experimental_results4}
\vspace*{-0.1in}
\end{figure}

\section{Discussion}
Our theoretical result focuses on the population limit case. We also explore via some experiments a more practically-relevant setting in which the sample size $n$ is finite and the optimal ML solution is the \emph{empirical} $r$-PCA. Specifically we compare WGAN-GP with the emprical $r$-PCA while varying a sample size $n\in\{1000, 10000, 100000, 1000000\}$. For the discriminator implementation in WGAN-GP, we employ a 5-layer neural network with 300 hidden units at each layer. See Fig.~\ref{fig:experimental_results4}. Observe that the WGAN-GP algorithm converges to the equilibrium with a non-significant gap to the emprical $r$-PCA performance, and the gap vanishes with an increase in $n$. From simulation results, we conjecture that WGAN optimization may achieve the empirical PCA solution for a finite sample-sized setting. The proof of the conjecture requires verification that the generator design w.r.t. the \emph{empirical} distribution of real data samples yields the empirical PCA solution. If that is not the case (the conjecture is disproved), then there may be a proper constraint for the discriminator that leads to the empirical PCA solution, as in the case of quadratic GAN~\cite{UGAN} where the 2nd-order Wasserstein distance together with a \emph{quadratic}-function constraint for the discriminator yields the empirical PCA solution.

\section{Conclusion}
We showed that WGAN can recover the optimal PCA solution under the linear-generator Gaussian-data setting. This promises that WGAN may form the basis of an optimal GAN architecture. Our future work includes: (1) Discovering a class of GAN architectures that performs PCA under the Gaussian setting (if any); (2) Exploring other settings in which other GAN architectures provide optimal solutions (if any); (3) Developing novel algorithms which can resolve stability issues (if required).

\section{Acknowledgment}
This work was supported by the National Research Foundation of Korea(NRF) grant funded by the Korea government(MSIT) (No. 2018R1A1A1A05022889).

%
%\section*{Acknowledgment}

%We are indebted to Michael Shell for maintaining and improving
%\texttt{IEEEtran.cls}. 

%%%%%%
%% To balance the columns at the last page of the paper use this
%% command:
%%
%\enlargethispage{-1.2cm} 
%%
%% If the balancing should occur in the middle of the references, use
%% the following trigger:
%%
%\IEEEtriggeratref{3}
%%
%% which triggers a \newpage (i.e., new column) just before the given
%% reference number. Note that you need to adapt this if you modify
%% the paper.  The "triggered" command can be changed if desired:
%%
%\IEEEtriggercmd{\enlargethispage{-20cm}}
%%
%%%%%%
%\newpage
\bibliographystyle{ieeetr}
\bibliography{bib_wgan}

\end{document}